\documentclass[conference]{IEEEtran}
\IEEEoverridecommandlockouts

\usepackage{amsmath,amssymb,bm,mathtools}
\usepackage{amsthm}
\usepackage{graphicx}
\usepackage{epstopdf}
\usepackage{booktabs}
\usepackage{cite}
\usepackage[bookmarks=false]{hyperref}
\usepackage{cleveref}
\usepackage{comment}
\graphicspath{{figs/}}
\usepackage{flushend}

\hypersetup{
    colorlinks = true,
    linkcolor = black,
    urlcolor = black,
    citecolor = black
}

\newtheoremstyle{resultcolon}{3pt}{3pt}{\normalfont}{}{\bfseries}{}{.5em}{\thmname{#1}\thmnumber{ #2}\thmnote{: #3}}
\newtheoremstyle{remarkcolon}{3pt}{3pt}{\normalfont}{}{\itshape}{:}{.5em}{}
\theoremstyle{resultcolon}
\newtheorem{proposition}{Proposition}

\theoremstyle{remarkcolon}
\newtheorem{remark}{Remark}

\title{Delay-Doppler Sensing Performance Analysis for MIMO-OFDM ISAC Systems \vspace{-0.1 cm}
}

\author{
    \IEEEauthorblockN{Peishi Li$^{\dag}$, Rang Liu$^{\ddag}$, Qian Liu$^{\dag}$, and Ming Li$^{\dag}$\\}
	\IEEEauthorblockA{$^{\dag}$ Dalian University of Technology, Dalian, Liaoning 116024, China \\ E-mail: \texttt{lipeishi@mail.dlut.edu.cn, \{mli, qianliu\}@dlut.edu.cn}}
	\IEEEauthorblockA{$^{\ddag}$ Friedrich-Alexander-University Erlangen-Nuremberg (FAU), Erlangen 91058, Germany \\ 
		E-mail: \texttt{rang.liu@fau.de}} \vspace{-0.4 cm}  }

\begin{document}

\maketitle
\thispagestyle{empty}
\pagestyle{empty}

\begin{abstract}
    In communication-centric integrated sensing and communication (ISAC), delay-Doppler sensing reuses data-bearing orthogonal frequency division multiplexing (OFDM) signals rather than dedicated radar probing waveforms. Consequently, the resulting range-Doppler map (RDM) is shaped not only by target parameters, but also by communication-symbol randomness and, in multi-antenna transmissions, by spatial beamforming. While existing analyses have largely focused on single-antenna OFDM-ISAC, the delay-Doppler sensing behavior of multi-antenna OFDM-ISAC remains insufficiently understood. This paper analyzes a multi-input multi-output (MIMO)-OFDM-ISAC system in which multiple data streams jointly illuminate a sensing target. We derive second-order moment expressions for the RDM under matched filtering (MF) and reciprocal filtering (RF), and use them to characterize the dynamic range (DR). The analysis reveals two key multi-stream effects. First, under MF, the random superposition of multiple beamformed streams creates an additional RDM floor beyond the modulation-dependent and receiver-noise terms; therefore, constant-modulus signaling no longer eliminates the data-induced pedestal as in single-antenna OFDM-ISAC. Second, under RF, the matched-angle data-induced floor is removed, but the noise floor is amplified according to the reciprocal-power statistics of the beamformed target illumination. These results show that the user-target angular geometry directly governs the MF/RF tradeoff: MF is more robust under weak illumination, whereas RF can achieve a higher DR when reciprocal-noise amplification is mild. Numerical results validate the analysis and demonstrate the distinct geometry-dependent behaviors of MF and RF.
\end{abstract}

\begin{IEEEkeywords}
    Integrated sensing and communication (ISAC), orthogonal frequency division multiplexing (OFDM), range-Doppler map, temporal-frequency filtering, dynamic range.
\end{IEEEkeywords}

\section{Introduction}

Integrated sensing and communication (ISAC) enables wireless networks to support data transmission and environmental sensing within a unified platform by sharing spectrum, hardware, and signal processing resources \cite{LiuFan_JSAC_2026,LiuRang_PROCIEEE_2026}. Among various ISAC architectures, communication-centric signaling is particularly attractive, as it directly reuses data-bearing communication waveforms for sensing and thus preserves compatibility with existing wireless standards. This reuse, however, makes the sensing waveform random rather than deterministic. In orthogonal frequency division multiplexing (OFDM)-based ISAC, the delay-Doppler response, or equivalently the range-Doppler map (RDM), is therefore affected not only by target parameters, but also by the statistics of the information-bearing symbols and the adopted temporal-frequency filtering operation \cite{LiuFan_JSAC_2026,LiuFan_TIT_2025}.

The impact of data randomness has recently been studied for single-antenna OFDM-ISAC, where constellation moments, pulse shaping, and filtering operations have been shown to determine the average sidelobe pedestal and delay-Doppler response \cite{LiuFan_TIT_2025,DuZhen_TWC_2026}. Related studies have further considered more general propagation conditions, such as targets with delays beyond the cyclic prefix (CP) duration \cite{LiPeishi_GLOBECOM_2025,LiPeishi_TSP_2026}. However, these results do not directly characterize multi-antenna communication-centric ISAC. In a multi-input multi-output (MIMO)-OFDM-ISAC system, the sensing target is illuminated by a beamformed superposition of multi-user data symbols. Hence, the effective sensing waveform along the target direction is jointly shaped by symbol randomness and spatial beamforming across antennas, users, and subcarriers. Existing multi-antenna ISAC studies mainly focus on joint target-parameter estimation, range-Doppler sidelobe suppression, and spatial-domain designs, such as beampattern shaping, transmit covariance optimization, hybrid beamforming, and communication-sensing tradeoff balancing \cite{XiaoZichao_TSP_2024,LiPeishi_TWC_2025,LiuFan_TWC_2018,LiuXiao_TSP_2020,Keskin_TWC_2025}. How such random beamformed OFDM signals behave after delay-Doppler processing remains insufficiently understood.

Motivated by this gap, this paper analyzes the delay-Doppler sensing performance of a MIMO-OFDM ISAC system. We derive range-Doppler map (RDM) second-order moment expressions under matched filtering (MF) and reciprocal filtering (RF), and characterize the dynamic range (DR). The analysis reveals that MF and RF respond to multi-stream random illumination in fundamentally different ways. Under MF, the RDM floor consists of multi-stream-induced, modulation-dependent, and receiver-noise contributions, so constant-modulus signaling no longer removes the data-induced pedestal as in single-antenna OFDM-ISAC. Under RF, the matched-angle data-induced floor vanishes, but the receiver noise is amplified by reciprocal-power moments of the beamformed illumination. These results show that the MF/RF tradeoff is strongly geometry-dependent: MF is more robust under weak or cancellation-prone illumination, whereas RF can achieve a higher DR when reciprocal-noise amplification is mild. Numerical results validate the theoretical expressions and demonstrate how user-target angular geometry shapes the distinct robustness behaviors of MF and RF.

\section{Signal Model and Delay-Doppler Processing}
We consider a downlink monostatic MIMO-OFDM ISAC transceiver, where an $N_{\mathrm{t}}$-antenna transmit array sends $K$ data-bearing streams over an OFDM frame with $N$ subcarriers and $M$ OFDM symbols. The same OFDM frame is reused for both downlink communication and target sensing, and no dedicated sensing signals are inserted.

On the $(n,m)$-th resource element, the frequency-domain transmit vector is modeled as
\begin{equation}
    \mathbf{x}_{n,m} = \mathbf{W}_n \mathbf{s}_{n,m},
\end{equation}
where $\mathbf{W}_n = [\mathbf{w}_{n,1}, \ldots, \mathbf{w}_{n,K}] \in \mathbb{C}^{N_{\mathrm{t}}\times K}$ denotes the transmit beamforming matrix on subcarrier $n$, and $\mathbf{s}_{n,m} = [s_{n,m,1}, \ldots, s_{n,m,K}]^T \in \mathbb{C}^{K}$ collects the $K$ data symbols transmitted on this resource element. The beamforming matrices are treated as deterministic over the considered OFDM frame, while the data symbols are assumed to be independent across users, subcarriers, and OFDM symbols, with $\mathbb{E}\{\mathbf{s}_{n,m}\}=\mathbf{0}$ and $\mathbb{E}\{\mathbf{s}_{n,m}\mathbf{s}_{n,m}^{H}\}=\mathbf{I}_K$.
Each symbol is proper, i.e., $\mathbb{E}\{s_{n,m,k}^{2}\}=0$, and has a common fourth-order moment $\mathbb{E}\{|s_{n,m,k}|^4\}=\mu_4$. The parameter $\mu_4$ captures modulation-dependent randomness under unit average-power normalization. In particular, constant-modulus constellations satisfy $\mu_4=1$, whereas circular Gaussian signaling gives $\mu_4=2$. These assumptions include standard proper QAM and PSK constellations \cite{LiuFan_TIT_2025}.

Let $\mathbf{a}(\theta)\in\mathbb{C}^{N_{\mathrm{t}}}$ denote the transmit steering vector toward angle $\theta$. The beamformed illumination along direction $\theta$ on the $(n,m)$-th resource element is defined as
\begin{equation}
    b_{n,m}(\theta) \triangleq \mathbf{a}^{H}(\theta)\mathbf{x}_{n,m}.
\end{equation}

We consider a single point target with angle $\theta_0$, delay $\tau_0$, Doppler frequency $f_{\text{d},0}$, and reflection coefficient $\alpha_0\sim\mathcal{CN}(0,\sigma_{\alpha}^{2})$. The echo is collected by a co-located single-antenna sensing receiver. This receive configuration is adopted to isolate the transmit-side random illumination effect induced by multi-user beamformed OFDM transmissions, without introducing additional receive-array beamforming effects. Assuming that the target round-trip delay does not exceed the CP duration, the received echo after CP removal and fast Fourier transform (FFT) is given by
\begin{equation}\label{eq:echo_model}
    y_{n,m} = \alpha_0 e^{-\jmath 2\pi n\Delta_f\tau_0} e^{\jmath 2\pi m f_{\mathrm{d},0}T_{\mathrm{s}}} b_{n,m}(\theta_0) + z_{n,m},
\end{equation}
where $z_{n,m}\sim\mathcal{CN}(0,\sigma_z^2)$ is the receiver noise, $\sigma_z^2=N_0\Delta_f$, $N_0$ denotes the noise power spectral density, and $T_{\mathrm{s}}$ is the OFDM symbol duration including the CP. Thus, although the receiver has a single antenna, the target echo still depends on the target angle through the transmit-side beamformed illumination.

Since the target echo in \eqref{eq:echo_model} is multiplied by the random factor $b_{n,m}(\theta_0)$, delay-Doppler processing should account for the data-induced fluctuation of the effective sensing waveform. In a monostatic communication-centric ISAC system, the sensing processor has access to the transmitted data symbols and beamforming matrices. However, the target angle is generally unknown a priori. We therefore introduce a general angular hypothesis $\theta$ and construct a temporal-frequency filtering $q_{n,m}(\theta)$ from the hypothesized beamformed illumination.

For a given angular hypothesis $\theta$, the RDM obtained with $q_{n,m}(\theta)$ is
\begin{equation}\label{eq:general_rdm}
    \chi_{\theta}(l,\nu) = \frac{1}{\sqrt{MN}} \sum_{m=0}^{M-1} \sum_{n=0}^{N-1} y_{n,m}q_{n,m}(\theta) e^{\jmath \frac{2\pi}{N}nl} e^{-\jmath \frac{2\pi}{M}m\nu}.
\end{equation}
This paper considers two representative filters. The MF $q_{n,m}^{\mathrm{MF}}(\theta)=b_{n,m}^{\ast}(\theta)$ matches the received echo to the hypothesized beamformed illumination, whereas the RF $q_{n,m}^{\mathrm{RF}}(\theta)={1}/{b_{n,m}(\theta)}$ normalizes the received echo by the hypothesized illumination.

The difference between MF and RF can be seen by substituting \eqref{eq:echo_model} into \eqref{eq:general_rdm}. MF yields the target factor $b_{n,m}(\theta_0)b_{n,m}^{\ast}(\theta)$, which reduces to the random illumination power $|b_{n,m}(\theta_0)|^2$ at the matched angle. Hence, MF preserves illumination fluctuations in the RDM and may create a signal-induced floor. By contrast, RF yields the target factor $b_{n,m}(\theta_0)/b_{n,m}(\theta)$, which becomes one at the matched angle and therefore cancels the data-induced target fluctuation. The price is that the receiver noise is scaled by $1/b_{n,m}(\theta_0)$ at the matched angle, making RF sensitive to the reciprocal-power statistics of the beamformed illumination. Thus, MF avoids reciprocal noise amplification but retains a residual signal-induced floor, whereas RF removes the matched-angle target fluctuation at the cost of possible noise enhancement. The following section quantifies these effects through the RDM second-order moments for a general angular hypothesis $\theta$.

\section{Second-Order RDM Moment Analysis}
For a given angular hypothesis $\theta$, the RDM output $\chi_{\theta}(l,\nu)$ is a complex random variable determined by the target coefficient, the receiver noise, and the random beamformed data symbols. Since $\alpha_0$ and $z_{n,m}$ are zero mean and independent of the transmitted symbols, both MF and RF yield $\mathbb{E}\{\chi_{\theta}(l,\nu)\} = 0$. Hence, the delay-Doppler sensing behavior is characterized through the second-order moments derived below.

\subsection{MF Processing}
Under MF with angular hypothesis $\theta$, \eqref{eq:general_rdm} can be decomposed as
\begin{equation}\label{eq:mf_decomp}
    \chi_{\theta}^{\mathrm{MF}}(l,\nu) = \alpha_0 \Gamma_{\theta}^{\mathrm{MF}}(l,\nu) + w_{\theta}^{\mathrm{MF}}(l,\nu),
\end{equation}
where
\begin{subequations}
    \begin{align}
        \Gamma_{\theta}^{\mathrm{MF}}(l,\nu)     & = \frac{1}{\sqrt{MN}} \sum_{m=0}^{M-1} \sum_{n=0}^{N-1} u_{n,m}(\theta_0,\theta) e^{\jmath \phi_{n,m}},                                                                 \\
        \!\!\!\! w_{\theta}^{\mathrm{MF}}(l,\nu) & \!= \!\frac{1}{\sqrt{MN}}\!\! \sum_{m=0}^{M-1}\! \sum_{n=0}^{N-1}\!\! z_{n,m} b_{n,m}^{\ast}\!(\theta) e^{\jmath \frac{2\pi}{N}nl} e^{-\jmath \frac{2\pi}{M}m\nu}, \!\!
    \end{align}
\end{subequations}
and $\phi_{n,m} = \frac{2\pi}{N}n(l-\tilde{l}_0) + \frac{2\pi}{M}m(\tilde{\nu}_0-\nu)$, $\tilde{l}_0 = \tau_0 N\Delta_f$ denotes the normalized delay, $\tilde{\nu}_0 \triangleq f_{\mathrm{d},0}MT_{\mathrm{s}}$ denotes the normalized Doppler, and $u_{n,m}(\theta_0,\theta) = b_{n,m}(\theta_0) b_{n,m}^{\ast}(\theta)$.

To make the angular dependence explicit, define $\beta_{n,k}(\theta) \triangleq \mathbf{a}^{H}(\theta) \mathbf{w}_{n,k}$ so that $b_{n,m}(\theta) = \sum_{k=1}^{K} \beta_{n,k}(\theta) s_{n,m,k}$.
We further define
\begin{subequations}
    \label{eq:mf_params}
    \begin{align}
        \rho_n(\theta_0,\theta)   & \triangleq \mathbb{E}\bigl\{u_{n,m}(\theta_0,\theta)\bigr\} = \sum_{k=1}^{K} \beta_{n,k}(\theta_0) \beta_{n,k}^{\ast}(\theta), \\
        \kappa_n(\theta_0,\theta) & \triangleq \sum_{k=1}^{K} |\beta_{n,k}(\theta_0)|^2 |\beta_{n,k}(\theta)|^2.
    \end{align}
\end{subequations}
For notational simplicity, let $\rho_n(\theta) = \rho_n(\theta,\theta)$, $\kappa_n(\theta) = \kappa_n(\theta,\theta)$.
With the aid of these definitions, the following proposition quantifies the MF-based RDM second moment by decomposing it into a coherent delay-Doppler response and three delay-Doppler-independent background contributions induced by multi-stream superposition, modulation randomness, and receiver noise, respectively.

\begin{proposition} \label{prop:mf_rdm_second}
    The MF-based RDM second moment under an arbitrary angular hypothesis $\theta$ is
    \begin{equation}\label{eq:mf_second}
        \begin{aligned}
             & \mathbb{E}\bigl\{|\chi_{\theta}^{\mathrm{MF}}\!(l,\nu)|^2 \!\bigr\} \!=\! \frac{\sigma_{\alpha}^{2}}{MN} \biggl| \sum_{n=0}^{N-1}\!\! \rho_n\!(\theta_0,\! \theta) e^{\jmath \frac{2\pi}{N} n (l - \tilde{l}_0)} \biggr|^2 \!\! \bigl|D_{M} (\tilde{\nu}_0 \!-\! \nu)\bigr|^2 \\
             & \hspace{1.5cm} + P_{\mathrm{stream}}^{\mathrm{MF}}(\theta_0,\theta) + P_{\mathrm{mod}}^{\mathrm{MF}}(\theta_0,\theta) + P_{\mathrm{noise}}^{\mathrm{MF}}(\theta),
        \end{aligned}
    \end{equation}
    where the three background floor contributions are defined as
    \begin{subequations}\label{eq:mf_terms}
        \begin{align}
            P_{\mathrm{stream}}^{\mathrm{MF}}(\theta_0,\theta) & \triangleq \frac{\sigma_{\alpha}^{2}}{N} \sum_{n=0}^{N-1} \Bigl( \rho_n(\theta_0)\rho_n(\theta) - \kappa_n(\theta_0,\theta) \Bigr), \label{eq:mf_stream} \\
            P_{\mathrm{mod}}^{\mathrm{MF}}(\theta_0,\theta)    & \triangleq \frac{\sigma_{\alpha}^{2} (\mu_4 - 1)}{N} \sum_{n=0}^{N-1} \kappa_n(\theta_0,\theta), \label{eq:mf_mod}                                       \\
            P_{\mathrm{noise}}^{\mathrm{MF}}(\theta)           & \triangleq \frac{\sigma_z^2}{N} \sum_{n=0}^{N-1} \rho_n(\theta). \label{eq:mf_noise_contribution}
        \end{align}
    \end{subequations}
\end{proposition}

\noindent \emph{Proof:}
By expanding $|\Gamma_{\theta}^{\mathrm{MF}}(l,\nu)|^2$ and separating the diagonal and off-diagonal terms, we obtain
\begin{equation}\label{eq:mf_gamma_step1}
    \begin{aligned}
         & \mathbb{E}\bigl\{|\Gamma_{\theta}^{\mathrm{MF}}(l,\nu)|^2\bigr\} = \frac{1}{MN} \biggl| \sum_{m,n} \mathbb{E}\bigl\{u_{n,m}(\theta_0,\theta)\bigr\} e^{\jmath \phi_{n,m}} \biggr|^2 \\
         & \qquad + \frac{1}{MN} \sum_{m,n} \Bigl( \mathbb{E}\bigl\{|u_{n,m}(\theta_0,\theta)|^2\bigr\} - |\rho_n(\theta_0,\theta)|^2 \Bigr).
    \end{aligned}
\end{equation}
The first term represents coherent delay-Doppler accumulation, while the second term characterizes the stochastic pedestal induced by the random data symbols.
To evaluate the latter, using the fourth-order identity
\begin{equation}
    \mathbb{E}\{s_i s_j^{\ast} s_k s_l^{\ast}\} = \delta_{ij} \delta_{kl} + \delta_{il} \delta_{jk} + (\mu_4 - 2) \delta_{ij} \delta_{jk} \delta_{kl},
\end{equation}
we obtain
\begin{equation}\label{eq:mf_fourth}
    \begin{aligned}
        \mathbb{E}\bigl\{|u_{n,m}(\theta_0,\theta)|^2\bigr\}
         & = \rho_n(\theta_0)\rho_n(\theta) + |\rho_n(\theta_0,\theta)|^2 \\
         & \quad + (\mu_4 - 2) \kappa_n(\theta_0,\theta).
    \end{aligned}
\end{equation}
Substituting this expression and using the fact that $\rho_n(\theta_0,\theta)$ is invariant over the OFDM symbol index $m$, we obtain
\begin{equation}\label{eq:mf_gamma_second}
    \begin{aligned}
         & \mathbb{E}\bigl\{ \! |\Gamma_{\theta}^{\mathrm{MF}}\!(l,\!\nu)|^2 \!\bigr\} \!=\! \frac{1}{MN} \! \biggl| \sum_{n}\! \rho_n(\theta_0, \!\theta) e^{\jmath \frac{2\pi}{N} n (l - \tilde{l}_0)} \biggr|^2 \! \bigl|D_{M} \!(\tilde{\nu}_0 \!-\! \nu)\bigr|^2 \\
         & ~ + \!\frac{1}{N}\! \sum_{n=0}^{N-1}\! \Bigl( \rho_n(\theta_0)\rho_n(\theta) \!-\! \kappa_n(\theta_0,\theta) \Bigr) + \frac{\mu_4 \!-\! 1}{N} \!\sum_{n=0}^{N-1}\! \kappa_n(\theta_0,\theta),
    \end{aligned}
\end{equation}
where $D_M(x) \triangleq \sum_{m=0}^{M-1} e^{\jmath \frac{2\pi}{M} m x}$ is the Dirichlet kernel. Combining \eqref{eq:mf_decomp} and \eqref{eq:mf_gamma_second}, and using $\mathbb{E}\{|w_{\theta}^{\mathrm{MF}}(l,\nu)|^2\}=P_{\mathrm{noise}}^{\mathrm{MF}}(\theta)$, gives \eqref{eq:mf_second} and \eqref{eq:mf_terms}, completing the proof. \hfill $\blacksquare$
\smallskip

Proposition \ref{prop:mf_rdm_second} shows that the MF-based RDM contains a coherent mainlobe and three background contributions. The term $P_{\mathrm{stream}}^{\mathrm{MF}}(\theta_0,\theta)$ originates from the random superposition of multiple beamformed streams and remains present even for constant-modulus signaling. The term $P_{\mathrm{mod}}^{\mathrm{MF}}(\theta_0,\theta)$ depends on the fourth-order constellation moment and vanishes for constant-modulus constellations. The term $P_{\mathrm{noise}}^{\mathrm{MF}}(\theta)$ is the receiver-noise floor after MF weighting.

\subsection{RF Processing}
For RF under an angular hypothesis $\theta$, \eqref{eq:general_rdm} can be written as
\begin{equation}\label{eq:rf_decomp}
    \chi_{\theta}^{\mathrm{RF}}(l,\nu) = \alpha_0 \Gamma_{\theta}^{\mathrm{RF}}(l,\nu) + w_{\theta}^{\mathrm{RF}}(l,\nu),
\end{equation}
where
\begin{subequations}
    \begin{align}
        \Gamma_{\theta}^{\mathrm{RF}}(l,\nu) & = \frac{1}{\sqrt{MN}} \sum_{m=0}^{M-1} \sum_{n=0}^{N-1} \frac{b_{n,m}(\theta_0)}{b_{n,m}(\theta)} e^{\jmath \phi_{n,m}}, \label{eq:rf_gamma}                  \\
        w_{\theta}^{\mathrm{RF}}(l,\nu)      & = \frac{1}{\sqrt{MN}} \sum_{m=0}^{M-1} \sum_{n=0}^{N-1} \! \frac{z_{n,m}}{b_{n,m}(\theta)} e^{\jmath \frac{2\pi}{N}nl} e^{-\jmath \frac{2\pi}{M}m\nu}. \!\!\!
    \end{align}
\end{subequations}

For unregularized RF, we assume that the reciprocal terms are well defined and have finite second-order moments. Define $\mathbf{g}_n(\theta) \triangleq \mathbf{W}_n^{H} \mathbf{a}(\theta)$, so that $b_{n,m}(\theta) = \mathbf{g}_n^{H}(\theta) \mathbf{s}_{n,m}$. Moreover, define
\begin{equation}
    r_{n,m}(\theta_0,\theta) \triangleq \frac{b_{n,m}(\theta_0)}{b_{n,m}(\theta)} = \frac{\mathbf{g}_n^{H}(\theta_0)\mathbf{s}_{n,m}}{\mathbf{g}_n^{H}(\theta)\mathbf{s}_{n,m}}.
\end{equation}
To make the RF moments explicit, for discrete constellations let $p(\mathbf{s})$ denote the probability mass function of $\mathbf{s} \in \mathcal{S}^{K}$. The exact ratio moments are given by
\begin{subequations}
    \begin{align}
        \zeta_n(\theta_0,\theta) & \triangleq \mathbb{E}\bigl\{r_{n,m}(\theta_0,\theta)\bigr\} =  \sum_{\mathbf{s} \in \mathcal{S}^{K}} p(\mathbf{s}) \frac{\mathbf{g}_n^{H}(\theta_0)\mathbf{s}}{\mathbf{g}_n^{H}(\theta)\mathbf{s}}, \label{eq:zeta_discrete}                         \\
        \psi_n(\theta_0,\theta)  & \triangleq \mathbb{E}\bigl\{|r_{n,m}(\theta_0,\theta)|^2\bigr\} \!=\!\! \sum_{\mathbf{s} \in \mathcal{S}^{K}} \! p(\mathbf{s}) \frac{|\mathbf{g}_n^{H}(\theta_0)\mathbf{s}|^2}{|\mathbf{g}_n^{H}(\theta)\mathbf{s}|^2}. \!\! \label{eq:psi_discrete}
    \end{align}
\end{subequations}
Similarly, we define the reciprocal-power moment of the beamformed illumination as
\begin{equation}\label{eq:xi_discrete}
    \xi_n(\theta) \triangleq \mathbb{E}\biggl\{\frac{1}{|b_{n,m}(\theta)|^2}\biggr\} = \sum_{\mathbf{s} \in \mathcal{S}^{K}} p(\mathbf{s}) \frac{1}{|\mathbf{g}_n^{H}(\theta)\mathbf{s}|^2},
\end{equation}
provided that $\xi_n(\theta)<\infty$, this condition ensures that the unregularized reciprocal filter does not encounter zero-illumination events for the considered discrete constellation. Then, the RF-based RDM second moment can be derived in the following proposition.

\begin{proposition} \label{prop:rf_rdm_second}
    The RF-based RDM second moment under an arbitrary angular hypothesis $\theta$ is
    \begin{equation}\label{eq:rf_second}
        \begin{aligned}
            \mathbb{E}\bigl\{|\chi_{\theta}^{\mathrm{RF}}(l,\nu)|^2 \!\bigr\}
             & \!=\! \frac{\sigma_{\alpha}^{2}}{MN} \! \biggl| \sum_{n=0}^{N-1} \!\zeta_n(\theta_0,\!\theta) e^{\jmath \frac{2\pi}{N} n (l - \tilde{l}_0)} \!\biggr|^2 \! \bigl|D_M(\tilde{\nu}_0 \!-\! \nu)\bigr|^2 \\
             & \qquad + P_{\mathrm{var}}^{\mathrm{RF}}(\theta_0,\theta) + P_{\mathrm{noise}}^{\mathrm{RF}}(\theta),
        \end{aligned}
    \end{equation}
    where the two background floor contributions are defined as
    \begin{subequations}\label{eq:rf_terms}
        \begin{align}
            P_{\mathrm{var}}^{\mathrm{RF}}(\theta_0,\theta) & \triangleq \frac{\sigma_{\alpha}^{2}}{N} \sum_{n=0}^{N-1} \Bigl( \psi_n(\theta_0,\theta) - |\zeta_n(\theta_0,\theta)|^2 \Bigr), \label{eq:rf_var} \\
            P_{\mathrm{noise}}^{\mathrm{RF}}(\theta)        & \triangleq \frac{\sigma_z^2}{N} \sum_{n=0}^{N-1} \xi_n(\theta). \label{eq:rf_noise_contribution}
        \end{align}
    \end{subequations}
\end{proposition}

\noindent \emph{Proof:}
Since $\mathbf{g}_n(\theta)$ depends only on the subcarrier index and the data symbols are independent over resource elements, the variables $\{r_{n,m}(\theta_0,\theta)\}$ are independent across $(n,m)$. Applying the same diagonal/off-diagonal decomposition as in the MF analysis gives
\begin{equation}\label{eq:rf_gamma_step1}
    \begin{aligned}
         & \mathbb{E}\bigl\{|\Gamma_{\theta}^{\mathrm{RF}}(l,\nu)|^2\bigr\} = \frac{1}{MN} \biggl| \sum_{m,n} \mathbb{E}\bigl\{r_{n,m}(\theta_0,\theta)\bigr\} e^{\jmath \phi_{n,m}} \biggr|^2 \\
         & ~~ + \!\frac{1}{MN}\! \sum_{m,n} \Bigl( \mathbb{E}\bigl\{|r_{n,m}(\theta_0,\theta)|^2\bigr\} \!-\! \bigl| \mathbb{E}\bigl\{r_{n,m}(\theta_0,\theta)\bigr\} \bigr|^2 \Bigr). \!\!
    \end{aligned}
\end{equation}
Substituting \eqref{eq:zeta_discrete} and \eqref{eq:psi_discrete} into \eqref{eq:rf_gamma_step1}, we obtain
\begin{equation}\label{eq:rf_gamma_second}
    \begin{aligned}
        \mathbb{E}\bigl\{|\Gamma_{\theta}^{\mathrm{RF}}(l,\nu)|^2\bigr\}
         & = \frac{1}{MN} \biggl| \sum_{m=0}^{M-1} \sum_{n=0}^{N-1} \zeta_n(\theta_0,\theta) e^{\jmath \phi_{n,m}} \biggr|^2     \\
         & \quad + \! \frac{1}{N} \!\sum_{n=0}^{N-1}\! \Bigl( \psi_n(\theta_0,\theta) \!-\! |\zeta_n(\theta_0,\theta)|^2 \Bigr).
    \end{aligned}
\end{equation}
Combining \eqref{eq:rf_decomp} and \eqref{eq:rf_gamma_second}, and using $\mathbb{E}\{|w_{\theta}^{\mathrm{RF}}(l,\nu)|^2\}=P_{\mathrm{noise}}^{\mathrm{RF}}(\theta)$, gives \eqref{eq:rf_second} and \eqref{eq:rf_terms}, completing the proof. \hfill $\blacksquare$
\smallskip

Compared with MF, RF replaces the beamformed-illumination product $u_{n,m}(\theta_0,\theta)$ by the ratio $r_{n,m}(\theta_0,\theta)$. Hence, $P_{\rm var}^{\rm RF}(\theta_0,\theta)$ is a ratio-variance floor, since $\psi_n(\theta_0,\theta)-|\zeta_n(\theta_0,\theta)|^2$ is the variance of $r_{n,m}(\theta_0,\theta)$ on subcarrier $n$. The term $P_{\rm noise}^{\rm RF}(\theta)$ is the reciprocal-filtered noise floor and is governed by the reciprocal-power moment $\xi_n(\theta)$, which becomes large when the hypothesized illumination $|b_{n,m}(\theta)|$ can be close to zero.

\subsection{Delay-Doppler Sensing Performance Analysis}
We now evaluate the RDM moments at the matched angular hypothesis $\theta=\theta_0$, which characterizes the delay-Doppler visibility of the target when the correct transmit-side direction is used. The mainlobe level is evaluated at the matched delay-Doppler point $(l,\nu)=(\tilde l_0,\tilde\nu_0)$. The background floor refers to the delay-Doppler-independent stochastic floor induced by data randomness and receiver noise, excluding deterministic Dirichlet sidelobes.

For MF, the mainlobe level can be expressed as
\begin{equation}\label{eq:mf_peak}
    P_{\mathrm{ml}}^{\mathrm{MF}} = \frac{\sigma_{\alpha}^{2} M}{N} \biggl| \sum_{n=0}^{N-1} \rho_n(\theta_0) \biggr|^2 + P_{\mathrm{floor}}^{\mathrm{MF}},
\end{equation}
and the background floor is
\begin{equation}\label{eq:mf_floor_matched}
    P_{\mathrm{floor}}^{\mathrm{MF}} = P_{\mathrm{stream}}^{\mathrm{MF}}(\theta_0,\theta_0) + P_{\mathrm{mod}}^{\mathrm{MF}}(\theta_0,\theta_0) + P_{\mathrm{noise}}^{\mathrm{MF}}(\theta_0).
\end{equation}
Accordingly, the MF DR can be written as
\begin{equation}\label{eq:dr_mf}
    \mathrm{DR}_{\mathrm{MF}} \triangleq \frac{P_{\mathrm{ml}}^{\mathrm{MF}}}{P_{\mathrm{floor}}^{\mathrm{MF}}} = 1 + \frac{\sigma_{\alpha}^{2} M \bigl| \sum_{n=0}^{N-1} \rho_n(\theta_0) \bigr|^2}{N P_{\mathrm{floor}}^{\mathrm{MF}}}.
\end{equation}
For fixed beamforming and noise power, $P_{\mathrm{floor}}^{\mathrm{MF}}$ is monotonically non-decreasing in $\mu_4$ because $\kappa_n(\theta_0)\ge 0$. Thus, constant-modulus constellations minimize the modulation-dependent MF floor, whereas constellations with larger fourth-order moments increase it. In particular, among commonly considered PSK/QAM/Gaussian signaling models, PSK yields the smallest modulation-dependent MF contribution, while circular Gaussian signaling yields a larger one.

For RF, at the matched angle $\theta = \theta_0$, one has $\zeta_n(\theta_0,\theta_0) = 1$ and $\psi_n(\theta_0,\theta_0) = 1$, so that the ratio-variance contribution vanishes. The RF mainlobe level is given by
\begin{equation}\label{eq:rf_peak}
    P_{\mathrm{ml}}^{\mathrm{RF}} = \sigma_{\alpha}^{2} MN + P_{\mathrm{floor}}^{\mathrm{RF}},
\end{equation}
and the RF background floor is
\begin{equation}\label{eq:rf_floor_matched}
    P_{\mathrm{floor}}^{\mathrm{RF}} = \frac{\sigma_z^2}{N} \sum_{n=0}^{N-1} \xi_n(\theta_0).
\end{equation}
The corresponding RF DR is
\begin{equation}\label{eq:dr_rf}
    \mathrm{DR}_{\mathrm{RF}} \triangleq \frac{P_{\mathrm{ml}}^{\mathrm{RF}}}{P_{\mathrm{floor}}^{\mathrm{RF}}} = 1 + \frac{\sigma_{\alpha}^{2} MN}{\frac{\sigma_z^2}{N} \sum_{n=0}^{N-1} \xi_n(\theta_0)}.
\end{equation}

\begin{remark}
    Constant-modulus symbols do not imply constant target illumination in MIMO-OFDM ISAC systems. Even when every stream uses a constant-modulus constellation, the effective illumination seen by the target is $b_{n,m}(\theta_0) = \sum_{k=1}^{K}\beta_{n,k}(\theta_0)s_{n,m,k}$, which is a random beam-domain phasor sum. Its magnitude depends not only on the symbol phases, but also on the user-target angular geometry through the beam gains $\{\beta_{n,k}(\theta_0)\}$. Therefore, constant-modulus signaling removes the modulation-dependent MF term $P_{\mathrm{mod}}^{\mathrm{MF}}$, but in general does not remove the multi-stream MF floor $P_{\mathrm{stream}}^{\mathrm{MF}}(\theta_0,\theta_0)$. Meanwhile, RF eliminates the matched-angle ratio-variance floor, but its noise floor is governed by the reciprocal-power moment $\xi_n(\theta_0)=\mathbb{E}\{|b_{n,m}(\theta_0)|^{-2}\}$, which is dominated by near-zero illumination events. Hence, the decisive quantity is not the modulus of each data symbol, but the distribution of the beamformed illumination at the target. MF is limited by illumination fluctuations, whereas RF is limited by the lower tail of the illumination power.
    \hfill $\blacksquare$
\end{remark}
\smallskip

For arbitrary subcarrier-dependent beamforming matrices $\{\mathbf{W}_n\}$, \eqref{eq:dr_mf} and \eqref{eq:dr_rf} characterize the DRs. Further insight can be obtained in the common-beamforming case, where $\mathbf{W}_n=\mathbf{W}$ for all $n$. Then $\rho_n(\cdot)=\rho(\cdot)$, $\kappa_n(\cdot)=\kappa(\cdot)$, and $\xi_n(\cdot)=\xi(\cdot)$ are independent of $n$. Define the sensing SNR as $\gamma_{\mathrm{s}}\triangleq\sigma_{\alpha}^{2}/\sigma_z^2$ and the MF signal-induced floor coefficient as
\begin{equation}\label{eq:S_mf}
    S_{\mathrm{MF}}(\theta_0) \triangleq \rho^2(\theta_0) - \kappa(\theta_0) + (\mu_4 - 1)\kappa(\theta_0).
\end{equation}
The DR expressions in \eqref{eq:dr_mf} and \eqref{eq:dr_rf} then reduce to
\begin{subequations}\label{eq:dr_closed}
    \begin{align}
        \mathrm{DR}_{\mathrm{MF}} & = 1 + \frac{\sigma_{\alpha}^2 M N \rho^2(\theta_0)}{\sigma_{\alpha}^2 S_{\mathrm{MF}}(\theta_0) + \sigma_z^2 \rho(\theta_0)}, \label{eq:dr_closed_mf} \\
        \mathrm{DR}_{\mathrm{RF}} & = 1 + \gamma_{\mathrm{s}} \frac{M N}{\xi(\theta_0)}. \label{eq:dr_closed_rf}
    \end{align}
\end{subequations}
The MF expression depends on both coherent illumination and signal-induced self-interference, whereas the RF expression is governed only by the reciprocal-noise factor at the matched angle.
When $S_{\mathrm{MF}}(\theta_0)>0$, comparing \eqref{eq:dr_closed_mf} and \eqref{eq:dr_closed_rf} gives
\begin{equation}\label{eq:dr_compare}
    \mathrm{DR}_{\mathrm{RF}}\ge \mathrm{DR}_{\mathrm{MF}} \Longleftrightarrow \gamma_{\mathrm{s}} \ge \frac{\rho^2(\theta_0)\bigl(\xi(\theta_0)-1/\rho(\theta_0)\bigr)}{S_{\mathrm{MF}}(\theta_0)}.
\end{equation}
The right-hand side is nonnegative because
\begin{equation}
    \xi(\theta_0) = \mathbb{E}\{|b_{n,m}(\theta_0)|^{-2}\} \ge \frac{1}{\mathbb{E}\{|b_{n,m}(\theta_0)|^2\}} = \frac{1}{\rho(\theta_0)},
\end{equation}
where the inequality follows from Jensen's inequality. If $S_{\mathrm{MF}}(\theta_0)=0$, the crossing condition in \eqref{eq:dr_compare} is not applicable, and the two DRs should be compared directly using \eqref{eq:dr_closed}. The single-effective-stream constant-modulus case is one such degenerate case.

\begin{table}[!t]
    \footnotesize
    \centering
    \caption{Simulation Parameters} \label{Tab:parameters}
    \vspace{-2mm}
    \begin{tabular}{ccc}
        \toprule
        \textbf{Parameter}            & \textbf{Symbol}  & \textbf{Value} \\
        \midrule
        Carrier frequency             & $f_{\text{c}}$   & 28~GHz         \\
        Subcarrier spacing            & $\Delta_f$       & 120~kHz        \\
        Number of antennas            & $N_{\mathrm{t}}$ & 16             \\
        Number of subcarriers         & $N$              & 256            \\
        Number of OFDM symbols        & $M$              & 128            \\
        Number of communication users & $K$              & 4              \\
        Target range                  & $R_0$            & 200~m          \\
        Target velocity               & $v_0$            & 9.4~m/s        \\
        Receiver noise PSD            & $N_0$            & -174~dBm/Hz    \\
        \bottomrule
    \end{tabular} \vspace{-0.2 cm}
\end{table}

The high-SNR behavior follows from \eqref{eq:dr_closed_mf}. When $\sigma_{\alpha}^2 \gg \sigma_z^2$, the MF DR approaches the finite ceiling
\begin{equation}\label{eq:dr_high_snr}
    \mathrm{DR}_{\mathrm{MF}} \to 1 + \frac{M N \rho^2(\theta_0)}{S_{\mathrm{MF}}(\theta_0)}.
\end{equation}
This saturation occurs because both the coherent MF mainlobe and the signal-induced MF floor scale with $\sigma_{\alpha}^{2}$. In contrast, the RF floor is noise-driven; hence, as long as $\xi(\theta_0)$ is finite, $\mathrm{DR}_{\mathrm{RF}}$ continues to increase linearly with $\gamma_{\mathrm{s}}$.

In the noise-dominated regime, i.e., $\sigma_{\alpha}^2 \ll \sigma_z^2$, \eqref{eq:dr_closed} gives
\begin{equation}\label{eq:dr_low_snr}
    \mathrm{DR}_{\mathrm{MF}} \approx 1 + \gamma_{\mathrm{s}} M N \rho(\theta_0),  \qquad \mathrm{DR}_{\mathrm{RF}} = 1 + \gamma_{\mathrm{s}} \frac{M N}{\xi(\theta_0)}.
\end{equation}
The corresponding initial slopes are
\begin{equation}\label{eq:dr_slopes}
    \left. \frac{\partial \mathrm{DR}_{\mathrm{MF}}}{\partial \gamma_{\mathrm{s}}} \right|_{\gamma_{\mathrm{s}} = 0} = M N \rho(\theta_0),  \qquad \frac{\partial \mathrm{DR}_{\mathrm{RF}}}{\partial \gamma_{\mathrm{s}}} = \frac{M N}{\xi(\theta_0)}.
\end{equation}
Since $\xi(\theta_0)\ge 1/\rho(\theta_0)$, the initial DR slope under MF is no smaller than that under RF. Thus, at low sensing SNR, RF does not yet benefit from eliminating the signal-induced floor but already pays the reciprocal-noise penalty.

\begin{remark}
    The common-beamforming expressions reveal that the MF/RF tradeoff is governed by three statistics of the same random illumination $b_{n,m}(\theta_0)$. The average illumination power $\rho(\theta_0)$ determines the coherent MF gain and the low-SNR MF slope. The coefficient $S_{\mathrm{MF}}(\theta_0)$ measures the signal-induced MF floor caused by multi-stream superposition and modulation randomness. The reciprocal-power moment $\xi(\theta_0)$ measures the sensitivity of RF to the lower tail of the illumination power. Hence, MF is more robust in the noise-dominated regime and under weak or cancellation-prone illumination, whereas RF becomes advantageous only when the target return is sufficiently strong and near-zero illumination events are rare. Since user-target angular geometry affects $\rho(\theta_0)$, $S_{\mathrm{MF}}(\theta_0)$, and $\xi(\theta_0)$ in different and generally non-monotonic ways, there is no universal MF/RF ordering over arbitrary geometries.
    \hfill $\blacksquare$
\end{remark}

\section{Numerical Results}
Table~\ref{Tab:parameters} summarizes the main simulation parameters. Unless otherwise specified, curves denote theoretical results and markers denote Monte Carlo simulations. The target angle is fixed at $\theta_0 = 10^\circ$. The transmit beams are generated by an ISAC-aware maximum-ratio transmission (MRT) rule. Specifically, the $k$-th data-stream beam is obtained by combining the MRT steering direction of the $k$-th user with the target steering direction.

\begin{figure}[!t]
    \centering
    \includegraphics[width=0.95\columnwidth]{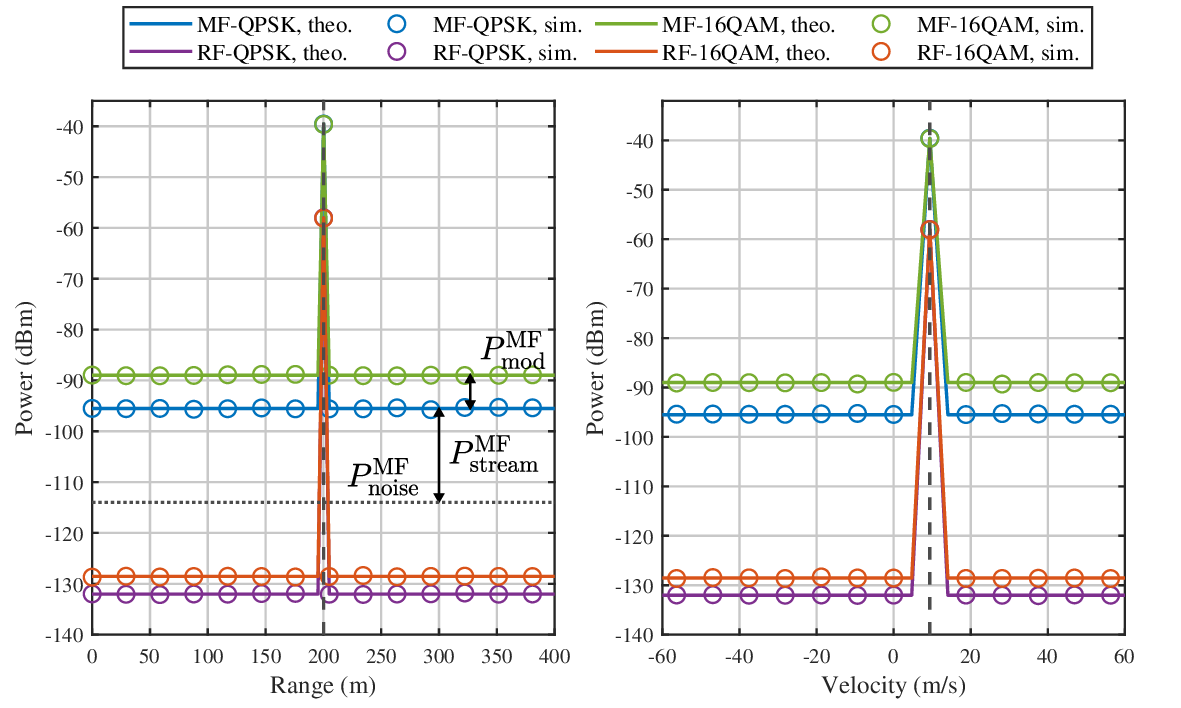}
    \vspace{-3mm}
    \caption{Range and velocity profiles with sensing SNR $\gamma_{\mathrm{s}}=20$ dB.}
    \vspace{-3mm}
    \label{fig:rdm_profile}
\end{figure}

Fig.~\ref{fig:rdm_profile} validates the derived second-order moment expressions by comparing the theoretical and simulated range/velocity profiles. The close agreement between the curves and markers confirms the accuracy of the analytical RDM characterization. The figure also illustrates the multi-stream effect highlighted in Remark~1. In contrast to single-antenna OFDM-ISAC, the MF-based RDM in the multi-antenna setting exhibits an additional pedestal caused by the multi-stream contribution $P_{\mathrm{stream}}^{\mathrm{MF}}$. Hence, even with QPSK signaling, the MF sidelobe floor is visibly elevated. For 16QAM, the floor further increases due to the additional modulation-dependent contribution $P_{\mathrm{mod}}^{\mathrm{MF}}$. By comparison, at the matched angle, RF removes the data-induced fluctuation, so its background level is mainly governed by the RF noise contribution.

\begin{figure}[!t]
    \centering
    \includegraphics[width=0.8\columnwidth]{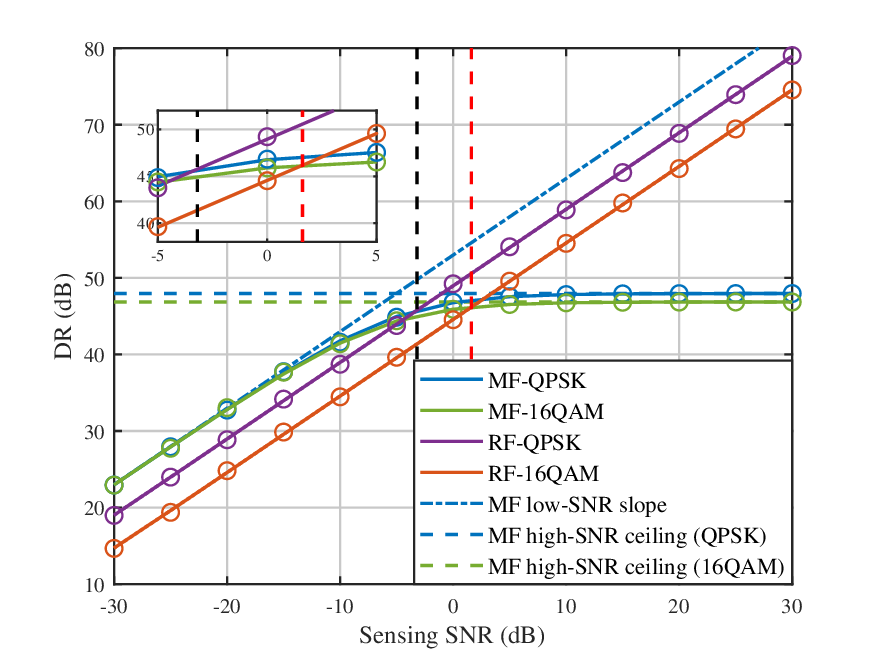}
    \vspace{-3mm}
    \caption{DR versus sensing SNR under MF and RF for QPSK and 16QAM constellations. The black and red vertical dashed lines mark the MF/RF crossing SNRs obtained by \eqref{eq:dr_compare} for QPSK and 16QAM, respectively.}
    \label{fig:dr_snr}\vspace{-4mm}
\end{figure}

\begin{figure}[!t]
    \centering
    \includegraphics[width=0.8\columnwidth]{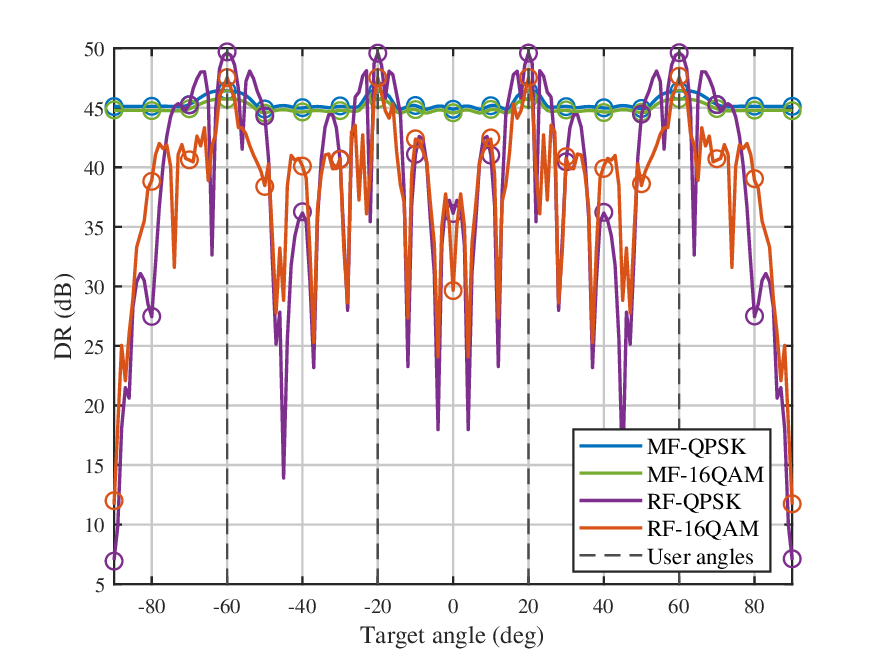}
    \vspace{-3mm}
    \caption{DR versus target angle under MF and RF for QPSK and 16QAM constellations, where the sensing SNR $\gamma_{\mathrm{s}}=0$~dB.}
    \label{fig:dr_angle}\vspace{-4mm}
\end{figure}

Fig.~\ref{fig:dr_snr} shows the DR as a function of the sensing SNR. The predicted crossing SNRs are $-3.2$~dB for QPSK and $1.6$~dB for 16QAM, which match the simulated transitions. In the high-SNR regime, the MF curves approach the analytical ceiling because both the coherent mainlobe and the MF floor scale with $\sigma_{\alpha}^{2}$. By contrast, the RF curves continue to increase with the sensing SNR, since the matched-angle RF floor remains noise-limited. In the low-SNR regime, the MF curves follow the analytical slope, which confirms that the initial DR growth is controlled by the coherent illumination power $\rho(\theta_0)$. More importantly, this figure shows a clear departure from single-antenna OFDM-ISAC: even under QPSK, MF and RF no longer yield identical DRs. This is because the multi-stream contribution $P_{\mathrm{stream}}^{\mathrm{MF}}$ remains nonzero under MF, whereas RF is still shaped by the reciprocal-power moment $\xi(\theta_0)$ through the RF noise contribution.

Fig.~\ref{fig:dr_angle} illustrates the impact of the user-target angular geometry. The $K=4$ user angles are fixed at $\pm60^\circ$ and $\pm20^\circ$, while the target angle varies over $[-90^\circ, 90^\circ]$. The MF curves vary relatively smoothly with the target angle, indicating that MF is comparatively robust to geometry changes. In contrast, the RF curves fluctuate much more strongly, which confirms the pronounced geometry sensitivity induced by $\xi(\theta_0)$. When the target direction is weakly illuminated by the transmit beams, small values of $|b_{n,m}(\theta_0)|$ occur more frequently. This sharply increases $\xi(\theta_0)$ and causes severe DR degradation under RF. An interesting observation is that RF-QPSK can even perform worse than RF-16QAM for some target angles. This does not contradict the MF ordering with respect to $\mu_4$, because RF is not determined by $\mu_4$ alone. Instead, RF depends on the full distribution of the beamformed illumination. For certain geometries, the superposition of constant-modulus QPSK phasors can create more severe near-zero illumination events than 16QAM, resulting in a larger reciprocal-power moment and hence a lower DR.

\section{Conclusion}
This paper analyzed the delay-Doppler sensing performance of a MIMO-OFDM ISAC system with multi-user beamformed data transmissions. Closed-form RDM second-order moment expressions were derived for MF and RF, and the corresponding DRs were characterized. The analysis showed that MF is limited by multi-stream-induced, modulation-dependent, and receiver-noise contributions, whereas RF removes the matched-angle data-induced floor but becomes sensitive to reciprocal-power statistics of illumination. These results reveal that user-target angular geometry directly shapes delay-Doppler target visibility in MIMO-OFDM ISAC systems. Numerical results validated the analysis and showed that MF is comparatively robust, while RF can be highly sensitive to weak or cancellation-prone target illumination.

\end{document}